\documentclass[11pt]{article}
\usepackage[a4paper,margin=1in]{geometry}
\usepackage[numbers,square]{natbib}
\usepackage{hyperref}
\usepackage{amsmath,amssymb,amsthm}
\usepackage{tabularx}

\theoremstyle{plain}
\newtheorem{theorem}{Theorem}
\newtheorem{lemma}[theorem]{Lemma}

\theoremstyle{definition}

\makeatletter
\def\@endtheorem{\endtrivlist}
\makeatother

\newcounter{rrule}
\renewcommand{\therrule}{R\arabic{rrule}}
\newenvironment{rrule}{\refstepcounter{rrule}\par\smallskip\noindent
\textbf{(\therrule)}\quad}{}

\newcounter{brule}
\renewcommand{\thebrule}{B\arabic{brule}}
\newenvironment{brule}{\refstepcounter{brule}\par\smallskip\noindent
\textbf{(\thebrule)}\quad}{}

\newcommand{\problem}[4]{%
\par\addvspace{\smallskipamount}
\begingroup
\setlength{\fboxsep}{4pt}
\noindent\fbox{%
\begin{tabularx}{\dimexpr\linewidth-2\fboxsep-2\fboxrule\relax}{@{}l@{\quad}X@{}}
\multicolumn{2}{@{}l@{}}{\textsc{#1} (#2)} \\[0.4ex]
\textbf{Input:} & #3 \\[0.2ex]
\textbf{Question:} & #4
\end{tabularx}%
}\par
\endgroup
\addvspace{\smallskipamount}
}

\newcommand{\edge}[2]{#1#2}
\newcommand{\degb}{\deg^{*}}
\newcommand{\Nb}{N^{*}}

\newcommand{\explode}[2]{#1_{#2}}
\newcommand{\I}[1]{I_{#1}}

\begin{document}
\title{Pathwidth-One Vertex Explosion and Co-Path Problems}
\author{Dekel Tsur\thanks{Stein Faculty of Computer and Information Science,
Ben-Gurion University of the Negev.}}
\date{}
\maketitle

\begin{abstract}
In the \textsc{Bipartite One-sided Vertex Explosion} problem (BOVE),
the input is a bipartite graph $G = (T, B, E)$ and a nonnegative integer $k$.
The goal is to decide whether there is a set $S \subseteq B$ of size at most $k$ such that
exploding every vertex in $S$ results in a graph with pathwidth at most one.
The restricted \textsc{Pathwidth One Vertex Explosion} problem (restricted POVE)
is a generalization of BOVE\@.
The input to this problem is a graph $G$, a nonnegative integer $k$, and a set of vertices $X \subseteq V(G)$.
The goal is to decide whether there is a set $S \subseteq X$ of size at most $k$ such that
exploding every vertex in $S$ results in a graph with pathwidth at most one.
In this paper, we show parameter-preserving reductions in both directions between BOVE and Co-Path Set,
and between restricted POVE and restricted Co-Path Packing.
These reductions yield improved parameterized algorithms for BOVE and restricted POVE
and an improved kernel for BOVE.
\end{abstract}


\section{Introduction}

In this paper, we consider the following graph modification problems:

\problem{Bipartite One-sided Vertex Explosion}{BOVE}
{A bipartite graph $G = (T, B, E)$ and a nonnegative integer $k$.}
{Is there a set $S \subseteq B$ of size at most $k$ such that
exploding every vertex in $S$ results in a graph with pathwidth at most one?}

\problem{Pathwidth One Vertex Explosion}{POVE}
{A graph $G$ and a nonnegative integer $k$.}
{Is there a set $S \subseteq V(G)$ of size at most $k$ such that
exploding every vertex in $S$ results in a graph with pathwidth at most one?}

\problem{Co-Path Set}{CPS}
{A graph $G$ and a nonnegative integer $k$.}
{Is there a set $S \subseteq E(G)$ of size at most $k$ such that
deleting the edges in $S$ results in a linear forest?}

\problem{Co-Path Packing}{CPP}
{A graph $G$ and a nonnegative integer $k$.}
{Is there a set $S \subseteq V(G)$ of size at most $k$ such that
deleting the vertices in $S$ results in a linear forest?}

Here, the \emph{vertex explosion} operation removes a vertex $v$ and, for each former neighbor $u$ of $v$,
adds a new vertex adjacent only to $u$.
A graph is a \emph{linear forest} if every connected component is a path.

In a \emph{restricted} variant of each of the problems above, the input also
includes a set $X$ of vertices or edges, and the set $S$ is required
to be a subset of $X$.
For restricted BOVE, we require $X \subseteq B$.

\citet{ahmed2023splitting} introduced the BOVE problem, which has applications in graph visualization,
and showed that BOVE is NP-complete.
\citet{ahmed2023fpt} gave a kernel for BOVE with $O(k^6)$ vertices.
\citet{baumann2024parameterized} considered the restricted POVE problem, which is
a generalization of the BOVE problem, and gave
an $O^*(4^k)$-time algorithm and a kernel with $O(k^2)$ vertices for this problem.
Recently, \citet{liu2026improved} gave an $O^*(2.31^k)$-time algorithm and
a kernel with $10k$ vertices for BOVE.

Parameterized algorithms and kernels for CPS were given
in~\citep{fernau2008parameterized,jiang2010weak,zhang2014radiation,feng2016kernelization,sullivan2017fast,li2017improved,Tsur_copath}.
The fastest parameterized algorithms for CPS 
are a deterministic $O^*(2^k)$-time algorithm given by \citet{Tsur_copath}
and a randomized $O^*(1.588^k)$-time algorithm given by \citet{sullivan2017fast}.
The smallest kernel for CPS is the kernel of \citet{li2017improved}, which has at most $4k$ vertices.

Parameterized algorithms for CPP were given
in~\citep{chen2010linear,feng2015randomized,Tsur_2bd,liu2026solving}.
The fastest parameterized algorithms for CPP 
are a deterministic $O^*(3^k)$-time algorithm given by \citet{Tsur_2bd}
and a randomized $O^*(2.925^k)$-time algorithm given by \citet{liu2026solving}.
In their concluding remarks, \citet{chen2010linear} stated that CPP has a kernel of size $O(k^2)$,
but did not provide a construction or proof.

In this paper, we give parameter-preserving reductions in both directions
between BOVE and CPS, and between POVE and CPP.
The reductions also apply to the restricted versions of these problems.
Consequently, we obtain the following results.
\begin{itemize}
\item
A deterministic $O^*(2^k)$-time algorithm and a randomized $O^*(1.588^k)$-time algorithm
for BOVE.

\item
A kernel for BOVE with at most $9k$ vertices.

\item
A deterministic $O^*(3^k)$-time algorithm for restricted POVE.
\end{itemize}

\section{Preliminaries}


A \emph{parameter-preserving reduction} from a parameterized problem $P$ to
a parameterized problem $Q$
is a polynomial-time algorithm that maps each instance $(I,k)$ of $P$
to an instance $(I',k')$ of $Q$ such that $k' \leq k$ and
$(I,k)$ is a yes-instance of $P$ if and only if $(I',k')$ is a yes-instance of $Q$.

Let $G$ be an undirected graph.
For a vertex $v \in V(G)$, $N_G(v)$ is the set of neighbors of $v$, $N_G[v] = N_G(v) \cup \{v\}$,
and $\deg_G(v) = |N_G(v)|$.
Additionally, $\Nb_G(v) = \{u \in N_G(v) : \deg_G(u) \geq 2 \}$ and
$\degb_G(v) = |\Nb_G(v)|$.
For a set of vertices $S \subseteq V(G)$, $N_G(S) = (\bigcup_{v \in S} N_G(v)) \setminus S$
and $N_G[S] = N_G(S) \cup S$.
For a set of vertices $S$, $G[S]$ is the graph $(S, \{\edge{u}{v} \in E(G) : u,v \in S\})$
and $G-S = G[V(G) \setminus S]$.
For a set of edges $S$, $G-S$ is the graph $(V(G), E(G) \setminus S)$.

A graph $G$ is a \emph{caterpillar forest} if removing every degree-1 vertex of $G$ results in
a linear forest.

A \emph{path decomposition} of a graph $G$
is a sequence of sets of vertices $X_1,\ldots,X_s$ such that
\begin{enumerate}
\item
Every vertex of $G$ belongs to at least one set $X_i$.
\item
For every edge $\edge{u}{v} \in E(G)$, there is at least one set $X_i$ such that $u,v \in X_i$.
\item
For every $i < j$, $X_i \cap X_j \subseteq X_l$ for every $i \leq l \leq j$.
\end{enumerate}
The \emph{width} of a path decomposition is $\max_i |X_i|-1$.
The \emph{pathwidth} of a graph $G$ is the minimum width of a path decomposition of $G$.

\begin{theorem}[\citet{arnborg1990monadic,baumann2024parameterized}]
Let $G$ be an undirected graph.
The following statements are equivalent.
\begin{enumerate}
\item
$G$ has pathwidth at most one.
\item
$G$ is a caterpillar forest.
\item
$G$ is acyclic and $\degb_G(v) \leq 2$ for every $v \in V(G)$.
\end{enumerate}
\end{theorem}

\section{Reductions}

For a graph $G$ and a set of vertices $S \subseteq V(G)$, let
$\explode{G}{S}$ be the graph obtained by exploding every vertex in $S$.

\begin{lemma}\label{lem:reduction-bove-cps}
There is a parameter-preserving reduction from BOVE to CPS\@.
Additionally, there is a parameter-preserving reduction from restricted BOVE to restricted CPS\@.
\end{lemma}
\begin{proof}
We first give a reduction between the unrestricted problems.
Given an instance $(G = (T, B, E),k)$ of BOVE, the reduction first
exhaustively applies the following reduction rules,
as formulated by \citet{liu2026improved} based on the rules of
\citet{ahmed2023fpt}.
We additionally remove isolated vertices in $B$, which is trivially safe.

\begin{rrule}\label{rule:B1-B3}
Let $B_1 = \{ v \in B : \deg_G(v) \leq 1 \}$ and
$B_3 = \{ v \in B : \degb_G(v) \geq 3 \}$.
Return $(G-(B_1 \cup B_3), k-|B_3|)$.
\end{rrule}

\begin{rrule}\label{rule:deg1-neighbor}
Let $v \in T$ be a degree-1 vertex whose unique neighbor has degree
at least 3.
Return $(G-\{v\},k)$.
\end{rrule}

After exhaustively applying these rules, denote the resulting instance by
$(G = (T, B, E),k)$.
We have that every vertex in $B$ has degree~2.
Construct a multigraph $G' = (T, E')$, where $E'$ is defined as follows.
For every vertex $v \in B$, add the edge $e_v = \edge{u_1}{u_2}$, where $u_1,u_2$ are the
two neighbors of $v$ in $G$.

We next show that $(G,k)$ is a yes-instance of BOVE if and only if $(G',k)$ is a yes-instance
of CPS\@.
Suppose that $(G,k)$ is a yes-instance of BOVE, and let $S$ be a solution of this instance.
Let $S' = \{e_v : v \in S\}$.
We claim that $G'-S'$ is a linear forest.
Suppose for contradiction that $G'-S'$ contains a cycle $u_1,u_2,\ldots,u_p,u_1$,
and let $e_{v_1},e_{v_2},\ldots,e_{v_p}$ be the edges of the cycle, where
$e_{v_i}$ is an edge between $u_i$ and $u_{i+1}$, with $u_{p+1} = u_1$.
Then $u_1,v_1,u_2,v_2,\ldots,u_p,v_p,u_1$ is a cycle in $\explode{G}{S}$, a contradiction.
Therefore, $G'-S'$ is acyclic.

Suppose for contradiction that $G'-S'$ contains a vertex $u$ with degree at least 3.
Let $e_{v_1},e_{v_2},e_{v_3}$ be three distinct edges incident to $u$ in $G'-S'$.
Then $v_1,v_2,v_3 \in \Nb_{\explode{G}{S}}(u)$, so
$\degb_{\explode{G}{S}}(u) \geq 3$, a contradiction.
Therefore, $G'-S'$ is a linear forest.
Thus, $S'$ is a solution for $(G',k)$ and $(G',k)$ is a yes-instance.

For the opposite direction, suppose that $(G',k)$ is a yes-instance and let $S'$
be a solution of this instance.
Let $S = \{v \in B : e_v \in S'\}$.
Create a graph $G_0$ from $G'-S'$ by subdividing each edge $e_v$ and naming the new vertex $v$.
Since $G'-S'$ is a linear forest, the graph $G_0$ is also a linear forest.
Let $S^*$ be the set of vertices that are created by exploding the vertices of $S$.
Then, $\explode{G}{S} - S^* = G_0$.
Therefore, $\explode{G}{S}$ is a caterpillar forest.
Thus, $S$ is a solution for $(G,k)$ and $(G,k)$ is a yes-instance.

The final step of the reduction is to exhaustively apply the following reduction rule
in order to remove parallel edges in $G'$.

\begin{rrule}\label{rule:parallel}
If $e_1$ and $e_2$ are parallel edges, return $(G' - \{e_1\},k-1)$.
\end{rrule}

Rule~\ref{rule:parallel} is safe because every solution must delete at least one of $e_1,e_2$.
At any stage of the reduction, if $k < 0$, return a no-instance.
After applying this rule, the reduction returns the current instance $(G',k)$.
This instance is equivalent to the input instance of BOVE\@.

For the restricted problems, we use the same construction with the following modifications.
Let $(G = (T, B, E),X,k)$ be the input instance of restricted BOVE.
Before each application of Rule~\ref{rule:B1-B3}, if $B_3 \nsubseteq X$, return a no-instance
of restricted CPS, since every vertex in $B_3$ must be exploded.
Otherwise, apply the rule and remove the deleted vertices from $X$.
Rule~\ref{rule:deg1-neighbor} is unchanged.
After exhaustively applying these rules, denote the resulting instance by
$(G = (T, B, E),X,k)$.
Construct $G'$ as above and let $X' = \{e_v : v \in X\}$.
The correspondence between $S$ and $S'$ above preserves the restrictions on the allowed sets.

Modify Rule~\ref{rule:parallel} as follows.
If neither of the parallel edges $e_1,e_2$ belongs to $X'$, return a no-instance of restricted CPS\@.
Otherwise, relabel the edges so that $e_1 \in X'$ and return
$(G' - \{e_1\},X' \setminus \{e_1\},k-1)$.
After exhaustively applying the modified Rule~\ref{rule:parallel}, return the current instance
$(G',X',k)$.
\end{proof}

\begin{lemma}\label{lem:reduction-cps-bove}
There is a parameter-preserving reduction from CPS to BOVE\@.
Additionally, there is a parameter-preserving reduction from restricted CPS to restricted BOVE\@.
\end{lemma}
\begin{proof}
We first give a reduction between the unrestricted problems.
Given an instance $(G,k)$ of CPS, the reduction returns $(G',k)$,
where $G' = (T, B, E')$ is a bipartite graph obtained from $G$ by subdividing each edge of $G$,
$T = V(G)$, and $B = V(G') \setminus T$.
The vertex obtained by subdividing an edge $e$ is denoted $b_e$.

Suppose that $(G,k)$ is a yes-instance of CPS, and let $S$ be a solution of this instance.
Let $S' = \{b_e : e \in S\}$.
Let $G_0$ be a graph obtained from $G-S$ by subdividing each edge of $G-S$.
Since $G-S$ is a linear forest, the graph $G_0$ is a linear forest.
Let $S^*$ be the set of vertices that are created by exploding the vertices of $S'$.
Then, $\explode{G'}{S'} - S^* = G_0$.
Hence, $\explode{G'}{S'}$ is a caterpillar forest.
Therefore, $S'$ is a solution of $(G',k)$ and $(G',k)$ is a yes-instance.

For the opposite direction, suppose that $(G',k)$ is a yes-instance,
and let $S'$ be a solution of this instance.
Let $S = \{e \in E(G) : b_e \in S'\}$.	
We claim that $G-S$ is a linear forest.
Suppose for contradiction that $G-S$ contains a cycle $v_1,v_2,\ldots,v_p,v_1$.
Then, $v_1,b_{\edge{v_1}{v_2}},v_2,b_{\edge{v_2}{v_3}},\ldots,v_p,b_{\edge{v_p}{v_1}},v_1$
is a cycle in $\explode{G'}{S'}$, a contradiction.
Therefore, $G-S$ is acyclic.

Suppose for contradiction that $G-S$ contains a vertex $v$ with degree at least 3.
Let $e_1,e_2,e_3$ be three distinct edges incident to $v$ in $G-S$.
Then $b_{e_1},b_{e_2},b_{e_3} \in \Nb_{\explode{G'}{S'}}(v)$, so
$\degb_{\explode{G'}{S'}}(v) \geq 3$, a contradiction.
Therefore, $G-S$ is a linear forest.
Thus, $S$ is a solution for $(G,k)$ and $(G,k)$ is a yes-instance.

The reduction between the restricted problems is as follows.
Given an instance $(G,X,k)$ of restricted CPS, the reduction returns $(G',X',k)$, where
$G'$ is the graph defined above and $X' = \{b_e: e \in X\}$.
\end{proof}

For a graph $G$, let $\I{G} = \{v \in V(G) : \deg_G(v) \geq 2 \}$.
Since exploding a vertex does not modify the degrees of the other vertices, we have
that every vertex in $V(G) \setminus S$ has the same degree in $G$ and $\explode{G}{S}$.
Additionally, every vertex created by the explosions has degree one in $\explode{G}{S}$.
Therefore, $\I{\explode{G}{S}} = \I{G} \setminus S$.
Additionally, $\explode{G}{S}[\I{\explode{G}{S}}] = G[\I{G}] - S$.

\begin{lemma}\label{lem:reduction-pove-cpp}
There is a parameter-preserving reduction from POVE to CPP\@.
Additionally, there is a parameter-preserving reduction from restricted POVE to restricted CPP\@.
\end{lemma}
\begin{proof}
We first give a reduction between the unrestricted problems.
Given an instance $(G,k)$ of POVE, the reduction returns $(G[\I{G}], k)$.

Suppose that $(G,k)$ is a yes-instance of POVE, and let $S$ be a solution of this instance.
Without loss of generality, $S \subseteq \I{G}$, since vertices of degree
at most one can be removed from $S$ without affecting the solution.
Since $\explode{G}{S}$ is a caterpillar forest,
$\explode{G}{S}[\I{\explode{G}{S}}] = \explode{G}{S}[\I{G} \setminus S]$ is a linear forest.
Therefore, $S$ is a solution for the CPP instance $(G[\I{G}], k)$
and this instance is a yes-instance.

For the opposite direction, suppose that $(G[\I{G}], k)$ is a yes-instance,
and let $S$ be a solution of this instance.
The graph $\explode{G}{S}$ is obtained from the graph $\explode{G}{S}[\I{\explode{G}{S}}] = G[\I{G}] - S$
by adding isolated vertices, isolated edges, and degree-1 vertices that are adjacent to vertices of
$G[\I{G}] - S$.
Since $G[\I{G}] - S$ is a linear forest, $\explode{G}{S}$ is a caterpillar forest.
Therefore, $S$ is a solution for the POVE instance $(G,k)$, and $(G,k)$ is a yes-instance.

The reduction between the restricted problems is as follows.
Given an instance $(G,X,k)$ of restricted POVE, the reduction returns $(G[\I{G}],X \cap \I{G}, k)$.
\end{proof}

\begin{lemma}\label{lem:reduction-cpp-pove}
There is a parameter-preserving reduction from CPP to POVE\@.
Additionally, there is a parameter-preserving reduction from restricted CPP to restricted POVE\@.
\end{lemma}
\begin{proof}
We first give a reduction between the unrestricted problems.
Given an instance $(G,k)$ of CPP, the reduction returns $(G',k)$,
where $G'$ is a graph obtained from $G$ by adding two private degree-1 neighbors $v_1,v_2$ for
every vertex $v \in V(G)$.
Note that $\I{G'} = V(G)$.

Suppose that $(G,k)$ is a yes-instance of CPP, and let $S$ be a solution of this instance.
The graph $\explode{G'}{S}$ is the graph obtained from $G-S$ by adding degree-1 vertices that
are adjacent to vertices in $G-S$ and adding paths of length one.
Since $G-S$ is a linear forest, $\explode{G'}{S}$ is a caterpillar forest.
Therefore, $S$ is a solution for the POVE instance $(G',k)$ and $(G',k)$ is a yes-instance.

For the opposite direction, suppose that $(G',k)$ is a yes-instance,
and let $S$ be a solution of this instance.
Without loss of generality, $S \subseteq \I{G'} = V(G)$, since vertices of degree
at most one can be removed from $S$ without affecting the solution.
We have $\I{\explode{G'}{S}} = \I{G'} \setminus S = V(G) \setminus S$ and
$\explode{G'}{S}[\I{\explode{G'}{S}}] = G-S$.
Since $\explode{G'}{S}$ is a caterpillar forest, $\explode{G'}{S}[\I{\explode{G'}{S}}] = G-S$ is a linear forest.
Therefore, $S$ is a solution for the CPP instance $(G,k)$ and $(G,k)$ is a yes-instance.

The reduction between the restricted problems is as follows.
Given an instance $(G,X,k)$ of restricted CPP, the reduction returns $(G',X,k)$, where
$G'$ is the graph defined above.
\end{proof}

\section{Parameterized algorithms and a kernel for BOVE}

From Lemma~\ref{lem:reduction-bove-cps} and the algorithms of \citet{Tsur_copath} and \citet{sullivan2017fast},
we obtain the following theorem.
\begin{theorem}
There is a deterministic $O^*(2^k)$-time algorithm and a randomized $O^*(1.588^k)$-time algorithm
for BOVE.
\end{theorem}

From Lemmas~\ref{lem:reduction-bove-cps} and~\ref{lem:reduction-cps-bove} and the kernel
of \citet{li2017improved}, we obtain the following theorem.
\begin{theorem}
There is a kernel for BOVE with at most $9k$ vertices.
\end{theorem}
\begin{proof}
Given an instance $(G,k)$ of BOVE, the kernel first applies the reduction of Lemma~\ref{lem:reduction-bove-cps}
to construct an instance $(G',k')$ of CPS, with $k' \leq k$.
Then, it applies the kernel of \citet{li2017improved} to $(G',k')$ to obtain an equivalent
instance $(G'',k'')$ such that $|V(G'')| \leq 4k''$ and $k'' \leq k'$.
If $(G'',k'')$ is a yes-instance and $S$ is a solution of $(G'',k'')$, then
$|E(G'')| = |S| + |E(G''-S)| \leq k'' + |V(G''-S)| \leq 5k''$.
Thus, if $|E(G'')| > 5k''$, the kernel returns some fixed no-instance.
Otherwise, the kernel applies the reduction of Lemma~\ref{lem:reduction-cps-bove} to $(G'',k'')$
to obtain an instance $(G''',k'')$ of BOVE.
By the proof of Lemma~\ref{lem:reduction-cps-bove},
$|V(G''')| = |V(G'')| + |E(G'')| \leq 9k'' \leq 9k$.
\end{proof}

\section{Parameterized algorithms for restricted POVE}

\begin{lemma}\label{lem:restricted-cpp-alg}
There is an $O^*(3^k)$-time algorithm for restricted CPP.
\end{lemma}
\begin{proof}
We use a recursive branching algorithm.
For an instance $(G,X,k)$, branching on a set $B \subseteq X$ means recursively solving
the instance $(G-B,X \setminus B,k-|B|)$.
If $k < 0$, the algorithm returns ``no''.
At each recursive call, the algorithm applies the first applicable rule.
The first two rules are Rules~\ref{rule:deg-3} and~\ref{rule:C4}, given below.
The remaining rules are the branching rules of \citet{Tsur_2bd}, in their
original order, with the modifications described below.
\begin{brule}\label{rule:deg-3}
Let $v$ be a degree-3 vertex such that $N_G[v] \nsubseteq X$.
If $N_G[v] \cap X \neq \emptyset$, branch on $\{u\}$ for every $u \in N_G[v] \cap X$.
Otherwise, return ``no''.
\end{brule}

\begin{brule}\label{rule:C4}
Let $C$ be the vertex set of a 4-cycle in $G$ such that $C \nsubseteq X$.
If $C \cap X \neq \emptyset$, branch on $\{u\}$ for every $u \in C \cap X$.
Otherwise, return ``no''.
\end{brule}

It is easy to verify that Rules~\ref{rule:deg-3} and~\ref{rule:C4} are correct.
The worst case of Rule~\ref{rule:deg-3} is when $|N_G[v] \cap X| = 3$.
In this case, the branching vector is $(1,1,1)$ and the branching number is 3.
The worst case of Rule~\ref{rule:C4} is when $|C \cap X| = 3$.
In this case, the branching vector is $(1,1,1)$ and the branching number is 3.

The branching rules of \citet{Tsur_2bd} are modified as follows:
In Rule~(1), the algorithm branches on some sets $B_1,\ldots,B_s$.
The modified rule branches on every $B_i$ that satisfies $B_i \subseteq X$.
The same modification is applied to Rule~(2).
We also correct two typographical errors in Rules~(6) and~(7) of \citet{Tsur_2bd}.
In Rule~(6), the neighbors of the degree-3 vertex must have degree at most two,
and in Rule~(7), the unique vertex in $N(v_1) \setminus \{v\}$ must have degree at most two.
Rules~(3)--(10) are otherwise unchanged.
Due to Rules~\ref{rule:deg-3} and~\ref{rule:C4}, each of these rules
branches only on subsets of $X$.
Although the correctness proofs of these rules in \citet{Tsur_2bd}
assume a solution of minimum cardinality,
their branch-containment and exchange arguments apply to arbitrary solutions.
Specifically, for each rule and every solution $S$ of a CPP instance $(G,k)$,
there is a branch $B$ such that either
\begin{enumerate}
\item
$B \subseteq S$, or
\item
there is a vertex $u \notin B$ and a vertex $v \in B$ such that
$S' = S \setminus \{u\} \cup \{v\}$ is a solution of $(G,k)$ and $B \subseteq S'$.
\end{enumerate}
Thus, if $S$ is a solution of an instance $(G,X,k)$ of restricted CPP, then there is
a branch $B$ of the rule such that either
\begin{enumerate}
\item
$B \subseteq S$, or
\item
there is a vertex $u \notin B$ and a vertex $v \in B$ such that
$S' = S \setminus \{u\} \cup \{v\}$ is a solution of $(G,X,k)$ (since $v \in B \subseteq X$)
and $B \subseteq S'$.
\end{enumerate}
Therefore, these rules are also correct for the restricted CPP problem.

\citet{Tsur_2bd} showed that when no rule of the algorithm can be applied,
each connected component of the graph is an induced path, an induced cycle,
the tree obtained from an edge by adding two degree-1 neighbors to each endpoint,
or the tree obtained from a three-vertex path by adding two degree-1 neighbors to each endpoint
and one degree-1 neighbor to the middle vertex.
For each component, compute the minimum number of allowed deletions, handling paths and cycles
directly and enumerating all allowed subsets for either type of tree.
Return ``yes'' if every component has a solution and the sum of these minima is at most $k$,
and ``no'' otherwise.
Therefore, such instances of restricted CPP can be solved in polynomial time.
The modifications do not increase the branching numbers of the original rules,
so the running time remains $O^*(3^k)$.
\end{proof}

From Lemmas~\ref{lem:reduction-pove-cpp} and~\ref{lem:restricted-cpp-alg}, we obtain the following theorem.

\begin{theorem}
There is a deterministic $O^*(3^k)$-time algorithm for restricted POVE.
\end{theorem}

\section*{Declaration of generative AI use}
During the preparation of this work, some results were obtained with the assistance of
ChatGPT (OpenAI).
The author independently verified each of these results.
The author wrote the entire first draft of the manuscript and subsequently used ChatGPT
to assist with revisions.
The author checked every edit made by ChatGPT for correctness and
takes full responsibility for the final manuscript.

\end{document}